\documentclass{article}

\usepackage{amsmath,amsthm,amssymb}
\usepackage[dvips]{graphicx}
\usepackage[dvips]{color}

\newcommand{\R}{\mathbf{R}}

\newcommand{\N}{\mathbf{N}}

\newtheorem{thm}{Theorem}[section]
\newtheorem{lem}[thm]{Lemma}
\newtheorem{prop}[thm]{Proposition}

\newtheorem{rem}[thm]{Remark}

\makeatletter

\@addtoreset{equation}{section}
\makeatother

\title
{Strong Convergence for Euler-Maruyama and Milstein Schemes with Asymptotic Method}

\author{
Hideyuki Tanaka
\\ 
{\small Graduate School of Science and Engineering, Ritsumeikan University }\\
{\small 1-1-1 Nojihigashi, Kusatsu, Shiga 525-8577, Japan }\\
{\small Email: hide.worldwide@gmail.com}
\\ \mbox{} \\
Toshihiro Yamada 
\\ 
{\small Graduate School of Economics, The University of Tokyo}\\
{\small 7-3-1 Hongo, Bunkyo, Tokyo, 113-0033, Japan}\\
{\small and } \\
{\small Mitsubishi UFJ Trust Investment Technology Institute Co., Ltd.\ (MTEC)}\\
{\small 2-6, Akasaka 4-Chome, Minato, Tokyo, 107-0052, Japan }\\
{\small Email: yamada@mtec-institute.co.jp}
}
\date{\empty}

\makeatletter
\def\section{\@startsection {section}{1}{\z@}{-3.5ex plus -1ex minus
-.2ex}{2.3ex plus .2ex}{\large\bf}}
\makeatother

\makeatletter
\def\subsection{\@startsection {subsection}{1}{\z@}{-3.5ex plus -1ex minus
-.2ex}{2.3ex plus .2ex}{\normalsize\bf}}
\makeatother

\begin{document}

\maketitle

\begin{abstract}
Motivated by weak convergence results in 
the paper of Takahashi and Yoshida (2005), 
we show strong convergence for an accelerated Euler-Maruyama scheme applied to perturbed stochastic differential equations. 
The Milstein scheme with the same acceleration is also discussed as an extended result. 
The theoretical results can be applied to analyzing the multi-level Monte Carlo method originally developed by M.B.\ Giles. 
Several numerical experiments for the SABR stochastic volatility model are presented in order to confirm the efficiency of the schemes.

\end{abstract}

\begin{flushleft}
\end{flushleft}

\section{Introduction}
We investigate an asymptotic method that accelerates numerical schemes for perturbed random variables. 
The general concept is as follows. 
Suppose that $F^\epsilon$ is a random variable depending on a small parameter $\epsilon$. 
Let us consider an approximation $\bar{F}^\epsilon$ for $F^\epsilon$ independently with respect to $\epsilon$. 
Then the bias $F^\epsilon - \bar{F}^\epsilon$ may be close to the bias 
$F^0 - \bar{F}^0$, since $\epsilon$ has a small effect on the value of $F^\epsilon - \bar{F}^\epsilon$. 
Therefore, we expect that 
\begin{equation}\label{asymptotic}
\bar{F}^\epsilon - \bar{F}^0 + F^0 \mbox{ is a better approximation than } \bar{F}^\epsilon.
\end{equation}
In particular, our interest is to study the above property when $F^\epsilon$ is 
a functional of a stochastic process and $\bar{F}^\epsilon$ comes from time discretization for it. 
In many cases, $F^0$ is a simpler model than $F^\epsilon$ and 
its exact distribution is well-known (e.g.\ Gaussian random variables or functionals of Gaussian processes). 
Even if the exact distribution of $F^0$ is unknown, it seems to be possible 
to provide a new scheme $\bar{F}^\epsilon - \bar{F}^0 + \tilde{F}^0$ with another more efficient scheme $\tilde{F}^0$ 
for $F^0$ (see Section \ref{subsec:sabr} as an example).

In general, we consider the following three error structures:
\begin{itemize}
\item Strong error:
\begin{align}\label{strong}
E[|F^\epsilon - (\bar{F}^\epsilon - \bar{F}^0 + F^0)|^p]^{1/p}.
\end{align}
\item Weak error:
\begin{align}\label{weak}
|E[F^\epsilon] - (E[\bar{F}^\epsilon] - E[\bar{F}^0] + E[F^0])|.
\end{align}
\item Monte Carlo bias estimator for  
$\frac{1}{M}\sum_{j=1}^M(F^{\epsilon,j} - F^{0,j}) + E[F^{0}]$ where $(F^{\epsilon,j})_j$ be an i.i.d.\ sampling of $F^\epsilon$: 
\begin{align}\label{CI}
\mathrm{Var}\Big(
\frac{1}{M}\sum_{j=1}^M (F^{\epsilon,j} - F^{0,j}) \Big).
\end{align}
\end{itemize}
Notice that in the case of Monte Carlo bias (\ref{CI}),  the term $\frac{1}{M}\sum_{j=1}^M (F^{0,j}) 
- E[F^0]$ works as a {\it control variates} method. 
For applications in strong error (\ref{strong}), we need an exact or accurate numerical simulation method for $F^0$. 
On the other hand, in the cases of weak error (\ref{weak}) and Monte Carlo bias (\ref{CI}), 
we have to know the value of $E[F^0]$, and therefore 
we need a closed formula or an accurate numerical scheme for $E[F^0]$ such as the fast Fourier transform in one dimension. 

When $F^\epsilon$ denotes a functional of a stochastic differential equation $X_t^\epsilon$, 
$\bar{F}^\epsilon$ corresponds to a certain time discretization scheme $\bar{X}_t^{\epsilon, (n)}$ ($n$: number of partition). 
Takahashi-Yoshida \cite{TY05} derived the following results in weak error sense (\ref{weak}) and 
Monte Carlo bias sense (\ref{CI}) for the Euler-Maruyama scheme $\bar{X}_t^{\epsilon, (n)}$:
\begin{align*}
E[f(X^\epsilon_T)] - ( E[f(\bar{X}_T^{\epsilon, (n)})] - E[f(\bar{X}_T^{0, (n)})] + E[f(X^0_T)] ) &= O\Big(\frac{\epsilon}{n}\Big),
\\ 
\mathrm{Var}^{1/2}\Big(
\frac{1}{M}\sum_{j=1}^M (f(\bar{X}_T^{\epsilon, (n), j}) - f(\bar{X}_T^{0, (n), j})) 
\Big) &= O\Big(\frac{\epsilon}{M^{1/2}}\Big),
\end{align*}
and hence the total error (the root-mean-squared error; RMSE) is equal to
\begin{align*}
 \mathrm{Var}^{1/2}\Big(E[f(X_T^{\epsilon})] 
- \frac{1}{M}\sum_{j=1}^M (f(\bar{X}_T^{\epsilon, (n), j}) - f(\bar{X}_T^{0, (n), j})) 
- E[f(X_T^{0})] \Big) 
\\  = O\Big(\frac{\epsilon}{n} + \frac{\epsilon}{M^{1/2}}\Big).
\end{align*}
Here they assumed some appropriate conditions for $f$ and the coefficients of $X_t^\epsilon$. 
This is the case where $F^\epsilon = f(X_T^\epsilon)$ and $\bar{F}^\epsilon = f(\bar{X}_T^{\epsilon, (n)})$ in (\ref{weak}), 
and $F^\epsilon = f(\bar{X}_T^{\epsilon, (n)})$ in (\ref{CI}). 
In order to make the total error $O(\gamma)$ with weak and Monte Carlo bias, 
the standard Euler-Maruyama scheme with i.i.d.\ sampling requires 
the computational cost $n \cdot M = O(\gamma^{-3})$, 
and in contrast, 
the accelerated Euler-Maruyama scheme with i.i.d.\ sampling requires 
the cost $O(\epsilon^3 \gamma^{-3})$. 
That is, the asymptotic method (\ref{asymptotic}) for the Euler-Maruyama scheme is 
$O(\epsilon^3)$-times faster than the standard method. 
Moreover, we can construct a sampling scheme whose computational cost turns out to be 
$O(\epsilon^{1-\delta} \gamma^{-2}(\log \gamma/\epsilon)^{2})$ for any $\delta>0$ and Lipschitz continuous function $f$ 
via the multi-level Monte Carlo method (See Theorem \ref{thm:mlmc}).

In this paper, we develop the error analysis for the Euler-Maruyama and Milstein schemes 
with the asymptotic method in strong sense (\ref{strong}). 
Under suitable conditions, we will show that for any $p \geq 2$, 
\begin{align}\label{MainClaim}
E\Big[\sup_{0\leq t \leq T}| X_{t}^\epsilon - ( \bar{X}_t^{\epsilon, (n)} - \bar{X}_t^{0, (n)} + X_t^{0}) |^p \Big]^{1/p} 
= O\Big( \frac{\epsilon}{n^\alpha}\Big)
\end{align}
with $\alpha = 1/2 \ ( = 1)$ for the Euler-Maruyama (Milstein, resp.) scheme $\bar{X}_t^{\epsilon, (n)}$. 
Although strong convergence is usually very slow, 
the asymptotic method (\ref{asymptotic}) helps to improve the speed of convergence. 

A simplest example of (\ref{MainClaim}) is for the case
where the SDE becomes the ODE when $\epsilon = 0$, namely, 
\[
dX_t^\epsilon = b(X_t^\epsilon)dt + \epsilon \sigma(X_t^\epsilon)dB_t.
\]
However, from the viewpoint of applications, we can also consider
the ($0$th-order) $\epsilon$-expansion around linear models like Black-Scholes 
(See an analytical expansion in Kunitomo-Takahashi \cite{KT03} and Takahashi-Yamada \cite{TY12}). 
Indeed, we can treat a perturbed stochastic differential equations such as 
\begin{align*}
dX_t^\epsilon &= b^\epsilon_tX_t^\epsilon dt + \sqrt{\sigma^\epsilon_t} X_t^\epsilon dB_t, 
\\ db^\epsilon_t &= h_b(b^\epsilon_t)dt + \epsilon V_b(b^\epsilon_t)dB_t, 
\\ d\sigma^\epsilon_t &= h_\sigma(\sigma^\epsilon_t)dt + \epsilon V_\sigma(\sigma^\epsilon_t)dB_t.
\end{align*}
Notice that $X_t^0$ becomes the Black-Scholes model with time-dependent coefficients. 
Therefore there are many applications in the models of 
dynamic assets with stochastic volatility and/or stochastic interest rate. 
In particular, we will discuss more general stochastic differential equations 
so-called local-stochastic volatility type models. 

This paper is organized as follows: 
Section 2 is devoted to state theoretical results for strong convergence (\ref{MainClaim}). 
In Section 3 we discuss pathwise simulation of stochastic volatility models. 
In Section 4, we introduce the multi-level Monte Carlo method and its acceleration by the asymptotic method. 
In Section 5 some numerical experiments for the SABR stochastic volatility model are given. 
In Appendix we give some mathematical results including the proof of main claims in Section 2 and 4.

\section{Strong convergence results}
As seen in the previous intruduction, 
the asymptotic method (\ref{asymptotic}) for discretizing stochastic processes 
is very natural to speed up the discretization procedure. 
We state here the basic setting to discuss the approximation schemes. 
Let us consider a stochastic differential equation (SDE) of the form
\begin{align}\label{sde}
d X_t^\epsilon = b(X_t^\epsilon, \epsilon)dt + \sigma(X_t^\epsilon, \epsilon)dB_t, \ \ X_0^\epsilon = x_0, 
\end{align}
where $b \in C(\R^N \times [0,1]; \R^N)$, $\sigma \in C(\R^N \times [0,1]; \R^N \times \R^d)$, and 
$B_t$ is a $d$-dimensional standard Brownian motion on a probability space $(\Omega, \mathcal{F}, P)$ 
with a filtration $(\mathcal{F}_t)_{t \geq 0}$ satisfying usual conditions. 
Throughout the paper, we use the equidistant partition $t_i = \frac{i}{n}T$, $0 \leq i \leq n$. 
The Euler-Maruyama and Milstein schemes will be considered with some smoothness conditions 
for the coefficients of the SDE.

\subsection{The Euler-Maruyama scheme with asymptotic method}
Let $\bar{X}_{t}^{\epsilon, (n)}$ be the Euler-Maruyama scheme for the SDE $X_t^\epsilon$ 
(Maruyama \cite{M55}): For $t \in [t_i, t_{i+1}]$, 
\begin{align}\label{EM}
\bar{X}_{t}^{\epsilon, (n)} := \bar{X}_{t_i}^{\epsilon, (n)}
 + b(\bar{X}_{t_i}^{\epsilon, (n)}, \epsilon)(t-t_i) + \sigma(\bar{X}_{t_i}^{\epsilon, (n)}, \epsilon) (B_t-B_{t_i}).
\end{align} 

The implementation of (\ref{EM}) is very simple. Indeed, practitioners only need to know 
how to simulate normal random variables. 
The error of the scheme has been analyzed deeply by many researchers (see e.g.\ \cite{TT90}, \cite{KP92}, \cite{BT95}). 
Roughly speaking, the strong order of convergence is equal to $1/2$, and the weak order is equal to $1$. 

We now prepare the assumptions for $\bar{X}$. 
\begin{description}
\item[($H_1$):] 
$
|b(x, \epsilon)| + |\sigma(x, \epsilon)| \leq C (1 + |x|).
$
\item[($H_2$):] 
$
|b(x, \epsilon) - b(y, \epsilon)| + |\sigma(x, \epsilon) - \sigma(y, \epsilon)| \leq C |x-y|.
$
\item[($H_3$):] 
$
|b(x, \epsilon) - b(x, 0)| + |\sigma(x, \epsilon) - \sigma(x, 0)| \leq C \epsilon (1 + |x|).
$
\item[($H_4$):] For every $\epsilon$, $b(\cdot, \epsilon), \sigma(\cdot, \epsilon) \in C^1$ and 
$
|\partial b(x, \epsilon) - \partial b(y, 0)| + |\partial \sigma(x, \epsilon) - \partial \sigma(y, 0)| 
\leq C (\epsilon + |x-y|).
$
\end{description}
The above constant $C$ is independent of $(x, y, \epsilon) \in \R^N\times \R^N \times [0,1]$.

Let us define the accelerated Euler-Maruyama scheme as 
$$\bar{Y}_t^{\epsilon, (n)} := \bar{X}_{t}^{\epsilon,(n)} - \bar{X}_{t}^{0, (n)} + X_{t}^0.$$
The property (\ref{asymptotic}) for strong convergence is formulated rigorously as follows. 
\begin{thm}\label{firsttheorem}
Suppose that $(H_1)$-$(H_4)$ hold.  
Then for any $p \geq 2$, there exists a constant $C = C(T,x_0,p)$ such that 
\begin{align*}
E\Big[\sup_{0\leq t \leq T}| X_{t}^\epsilon - \bar{Y}_t^{\epsilon, (n)} |^p \Big]^{1/p} \leq C \frac{\epsilon}{n^{1/2}}.
\end{align*}
\end{thm}

In particular, if we consider the small volatility model $dX_t^\epsilon = b(X_t^\epsilon)dt + \epsilon dB_t$ 
and the ODE $dX_t^0 = b(X_t^0)dt$, then intuitively speaking, 
$(X_t^\epsilon - \bar{X}_t^\epsilon) - (\bar{X}_t^0 - \bar{X}_t^0)$ cancels out 
the error from the drift term (except the effect of $\epsilon$), and the error from 
small volatility $\epsilon$ only remains. Hence the total error is proportional to $\epsilon$. 

Of course, more general situations can be considered, for example, 
if $b$ and $\sigma$ depends on time $t$, then some smoothness assumptions with respect to $(t, \epsilon)$ 
are needed in addition to $(H_1)$-$(H_4)$. 
We will not attempt to prove this, but basically the asymptotic method works as well. 

\begin{rem}
The rate of convergence of the Euler-Maruyama scheme basically 
relies on the smoothness (or the Lipschitz continuity) of coefficients of SDEs. 
If the coefficients are not smooth but H\"older continuous, 
the speed of convergence may be slow, 
as seen in the paper by Yan \cite{Yan02}. 
For obtaining the strong rate of convergence $O(n^{-1/2})$ with $\sigma(x) = x^\alpha$ ($1/2 \leq \alpha < 1$), 
a modified Euler-type scheme (called a symmetrized Euler scheme) was developed by Berkaoui et al.\ (\cite{BBD08}). 
We should mention that the Euler-Maruyama scheme may not converge strongly 
when the coefficients are non-globally Lipschitz continuous. 
For example, in \cite{HJK11} a sufficient condition that the scheme explodes is given. 
\end{rem}

\subsection{The Milstein scheme with asymptotic method}
We next discuss the Milstein scheme which has a higher order rate of convergence than 
the Euler-Maruyama scheme in strong sense. 
Just for notational convenience, we only consider the case $d=1$. 
Of course, in general dimensional setting with commutative vector fields $(\sigma^{j})_{1\leq j \leq d}$, 
we can use the (accelerated) Milstein scheme as well. 

Throughout this section, we assume the following smoothness. 
\begin{itemize}
\item For every $\epsilon$, $\sigma(\cdot, \epsilon) \in C^2.$
\end{itemize}
The Milstein scheme $\hat{X}_{t}^{\epsilon, (n)}$ for the SDE $X_t^\epsilon$ is defined by
\begin{align*}
\hat{X}_{t}^{\epsilon, (n)} &:= \hat{X}_{t_i}^{\epsilon, (n)}
 + b(\hat{X}_{t_i}^{\epsilon, (n)}, \epsilon)(t-t_i) + \sigma(\hat{X}_{t_i}^{\epsilon, (n)}, \epsilon) (B_t-B_{t_i}) 
 \\ & \ \ \ \ + \sigma\sigma'(\hat{X}_{t_i}^{\epsilon, (n)}, \epsilon) \int_{t_i}^t\int_{t_i}^s dB_r dB_s
 \\ &= \hat{X}_{t_i}^{\epsilon, (n)}
 + b(\hat{X}_{t_i}^{\epsilon, (n)}, \epsilon)(t-t_i) + \sigma(\hat{X}_{t_i}^{\epsilon, (n)}, \epsilon) (B_t-B_{t_i}) 
 \\ & \ \ \ \ + \frac{1}{2}\sigma\sigma'(\hat{X}_{t_i}^{\epsilon, (n)}, \epsilon) ((B_t-B_{t_i})^2 - (t-t_i))
\end{align*}
for $t \in [t_i, t_{i+1}]$.

We use the (stronger) assumptions for $\hat{X}$. 
\begin{description}
\item [($H'_1$):] $(H_1)$ \& 
$
|\sigma \sigma'(x, \epsilon)| 
+ |b \sigma'(x, \epsilon)| 
+ |\sigma^2 \sigma''(x, \epsilon)| 
\leq C (1 + |x|).
$
\item [($H'_2$):] $(H_2)$ \& 
$
|\sigma\sigma'(x, \epsilon) - \sigma\sigma'(y, \epsilon)| 
+ |b\sigma'(x, \epsilon) - b\sigma'(y, \epsilon)|
+ |\sigma^2\sigma''(x, \epsilon) - \sigma^2\sigma''(y, \epsilon)|
\leq C |x-y|.
$
\item [($H'_3$):] $(H_3)$ \& 
$
|\sigma\sigma'(x, \epsilon) - \sigma\sigma'(x, 0)| 
+ |b\sigma'(x, \epsilon) - b\sigma'(x, 0)| 
+ |\sigma^2\sigma''(x, \epsilon) - \sigma^2\sigma''(x, 0)| 
\leq C \epsilon (1 + |x|).
$
\item [($H'_4$):] $(H_4)$ \& 
$
|(\sigma\sigma')'(x, \epsilon) - (\sigma\sigma')'(y, 0)| 
\leq C (\epsilon + |x-y|).
$
\end{description}

Let us define the accelerated Milstein scheme as 
$$
\hat{Y}_t^{\epsilon, (n)} := \hat{X}_{t}^{\epsilon,(n)} - \hat{X}_{t}^{0, (n)} + X_{t}^0.
$$ 
Then we can get the higher order convergence rate. 
\begin{thm}\label{secondtheorem}
Suppose that $(H'_1)$-$(H'_4)$ hold. 
Then for any $p \geq 2$, there exists a constant $C = C(T,x_0,p)$ such that 
\begin{align*}
E\Big[\sup_{0\leq t \leq T}| X_{t}^\epsilon - \hat{Y}_t^{\epsilon, (n)} |^p \Big]^{1/p} \leq C \frac{\epsilon}{n}.
\end{align*}
\end{thm}

\section{Application to pathwise simulation of stochastic volatility models}
Our goal in this section is to construct a faster pathwise approximation 
for perturbed stochastic differential equations which appear in financial modeling of volatility. 

\subsection{An accelerated scheme for SABR model}\label{subsec:sabr}

In financial modeling, the SABR model plays a role to fit the implied volatility especially in short time. 
The model is given by the SDE (Hagan et al.\ \cite{HKLW02}) 
\begin{align*}
dS_t &= \sqrt{\alpha_t} S_t^\beta dB^1_t
\\ d \alpha_t &= \nu \alpha_t (\rho dB_t^1 + \sqrt{1-\rho^2} dB_t^2).
\end{align*}
The volatility is not a mean-reversion process, hence this model does not suit for pricing long-dated options. 
If $\beta < 1$, as far as the authors know, there is no exact pathwise simulation method for the above SDE. 
In weak sense, several accurate simulation methods via Bessel processes are known.

To avoid that the volatility process $\alpha_t$ becomes negative in approximation procedures, 
we use a logarithmic transform for $\alpha_t$. 
\begin{align*}
dS_t &= \sqrt{\alpha_0 \exp(\tilde{\alpha}_t)} S_t^\beta dB^1_t
\\ d \tilde{\alpha_t} &= -\frac{\nu^2}{2} dt + \nu (\rho dB_t^1 + \sqrt{1-\rho^2} dB_t^2).
\end{align*}

Consider $\epsilon = \nu$. Since we do not know exact pathwise simulation methods for $S_t^0$, 
we substitute the Milstein scheme $\hat{S}_t^0$ for $S_t^0$. 
Therefore, we can use an $O(\frac{\epsilon}{\sqrt{n}} + \frac{1}{n})$-scheme defined by 
\begin{align*}
\tilde{Y}_t^{\epsilon} := 
\bar{S}_t^\epsilon - \bar{S}_t^0 + \hat{S}_t^0.
\end{align*}
When $\nu$ is small enough, a typical sample path is like Figure \ref{fig1}. 
Here we use $n=16$ for the standard Euler-Maruyama scheme (Standard E-M) and the above accelerated scheme (Accelerated).

\begin{figure}[tb]
  \begin{center}
      \includegraphics[width=75mm]{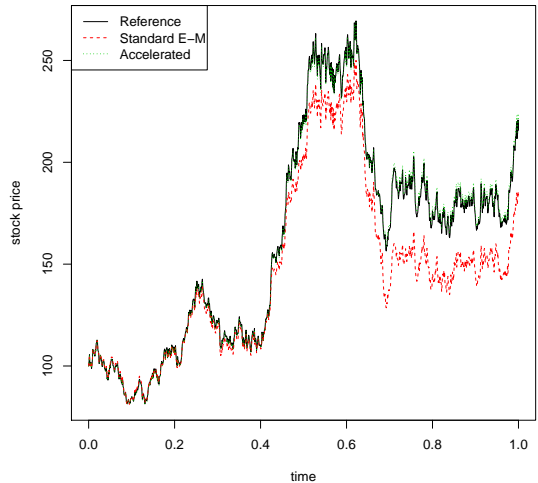} 
    \caption{A sample path of discretized SABR model when $\nu$ is small.}
    \label{fig1}
  \end{center}
\end{figure}

We next turn to consider another formal approximation scheme. 
Formally, $x^\beta \approx x$ when $\beta \approx 1$ and especially $x \approx 1$. 
Thus consider the scaling $L_t := S_t / S_0$, $\beta = \beta(\epsilon)$. 
\begin{align*}
dL_t^\epsilon &= \sqrt{\alpha_0 \exp(\tilde{\alpha}_t)}S_0^{\beta-1} (L_t^\epsilon)^\beta dB^1_t, 
\\ dL_t^0 &= \sqrt{\alpha_0}S_0^{\beta-1} L_t^0 dB^1_t.
\end{align*}
Here $S_0^{\beta-1}$ is just a constant coming from the scaling, 
thus we do not change the constant $S_0^{\beta-1}$ even when $\epsilon = 0$. 
The accelerated scheme that we want to use is 
\[
\check{Y}_t^\epsilon := \bar{S}_t^\epsilon - S_0( \bar{L}_t^0 - L_t^0 ).
\]
Since $L_t^0$ is a log-normal process, it is useful to compute the path $t \mapsto X_t^0$ and $E[f(X_T^0)]$.
We will check the efficiency of $\check{Y}$ through a numerical test later.

\subsection{General stochastic volatility models}
The following model is an extension of local-stochastic volatility models 
applicable to both short and long term contingent claims in financial markets.
\begin{align*}
dS_t &= \mu S_t dt + \sqrt{\alpha_t} S_t^\beta dB^1_t + S_{t-}dJ_t
\\ d \alpha_t &= \lambda(\theta - \alpha_t)dt + \nu \alpha_t^\gamma (\rho dB_t^1 + \sqrt{1-\rho^2} dB_t^2)
\end{align*}
where $J_t$ is a compound Poisson process, which is often used to adapt especially short-dated large volatility smile/skew. 

We remark that it is difficult to fit short-dated volatility smile/skew under the Heston model ($\beta=1$, $\gamma=1/2$, $J_t \equiv 0$), 
and then $\nu$ can take very large value. 
On the other hand, under general models with $\beta$ and $J_t$, the parameter $\nu$ need not to be so large. 

Let $\{\tau_j\}$ be the random jump times associated to $J_t$ and consider a new time partition 
$\{\tilde{t}_k\} := \{t_i\} \cup \{\tau_j\}$. On the time interval $[\tilde{t}_k, \tilde{t}_{k+1})$
we can regard the approximation problem for $S_t$ as the one for a continuous SDE. 
In particular by taking $\epsilon = \nu = 0$, the model becomes the CEV model with time-dependent coefficients.
For a technical reason, we should consider some carefull treatments 
around zero of the function $(\cdot)^\gamma$ (See \cite{BD04, BBD08, LDV10}).


\section{Application to multi-level Monte Carlo method}

The theoretical results we obtained in previous can be applied to the multi-level Monte Carlo method (MLMC in short). 
We propose an accelerated  Monte Carlo sampling for Takahashi-Yoshida's weak convergence method.

\subsection{The basic methodology of MLMC}
We forget the parameter $\epsilon$ for the time being, and denote by $X_t$ the continuous SDE $X_t^\epsilon$ defined by (\ref{sde}). 
Let us define $P := f(X_T)$ and $\bar{P}_l := f(\bar{X}_T^{(n_l)})$, and consider the time-step size 
$T/n_l = T / k^l$ for a fixed $k \in \N$. 
Let $L \in \N$ and the sampling of multi-level Monte Carlo is defined by 
\begin{align}\label{MLMC}
Y = \sum_{l=0}^L Y_l
\end{align}
where each $Y_l$ is independently distributed and is given by 
\begin{align*}
Y_l = \frac{1}{N_l} \sum_{j=1}^{N_l} 
\left\{
\begin{array}{ll}
\bar{P}_0^{(j)}, & \mbox{ if } l = 0,\\
( \bar{P}_l - \bar{P}_{l-1} )^{(j)}, & \mbox{ if } l \geq 1, 
\end{array}
\right.
\end{align*}
with i.i.d.\ sampling $\bar{P}_0^{(j)}$ or $( \bar{P}_l - \bar{P}_{l-1} )^{(j)}$, $j = 1, \dots, N_l$. 
The most important point is to use the same Brownian motion path $(B_t)_{t\geq 0}$ 
for simulating $\bar{P}_l$ and $\bar{P}_{l-1}$, and so the concept of the multi-level Monte Carlo method 
concerns the strong (pathwise) convergence rate. 

Clearly we show that 
\[
E[P] - E[Y] = E[P] - E[\bar{P}_L],
\]
therefore the weak rate of convergence depends only on the last number $L$. 
Moreover, we obtain from the independence of $(Y_l)$, 
\[
\mathrm{Var}(Y) = \sum_{l=0}^L \frac{1}{N_l} \mathrm{Var}(Y_l).
\]
and by definition $\mathrm{Var}(Y_l) = \mathrm{Var}(\bar{P}_{l} - \bar{P}_{l-1})$. 
Suppose some suitable conditions for $f$ and $X_t$. Then one can obtain 
\begin{align*}
E[\bar{P}_l - P] &= O(1/n_l), 
\\ \mathrm{Var}^{1/2}(\bar{P}_l - \bar{P}_{l-1}) &\leq \| \bar{P}_l - P\|_2 + \|\bar{P}_{l-1} - P \|_2
\\ &= O(1/ n_l^{1/2}). 
\end{align*}
The last estimate is the strong convergence result in $L^2$, 
which is discussed in this paper. 

The total computational cost $C$ is determined by the level $L$, the number of sampling $(N_l)_{l=0}^L$, and 
the number of partition $n_l (= k^l)$ so that 
\[
C = \sum_{l=0}^L N_l n_l.
\]
Suppose that the required RMSE is $O(\gamma)$. 
Then by choosing $N_l = O(\gamma^{-2} L n_l^{-1})$, the total variance $\mathrm{Var}(Y)$ is of $O(\gamma^2)$. 
Now if we set $L = \log(\gamma^{-1}) / \log(k) + O(1)$, 
then the total time discretization error $E[\bar{P}_L - P] = O(\gamma)$.
Consequently $C = O(\gamma^{-2}(\log \gamma )^2)$ for the required accuracy $O(\gamma)$.

\subsection{Accelerated MLMC sampling (with smooth payoffs)}
From now on, we reconsider the sampling of the accelerated Euler-Maruyama scheme 
introduced by Takahashi and Yoshida from the standard Monte Carlo method 
\[
\frac{1}{M}\sum_{j=1}^M (f(\bar{X}_T^{\epsilon, (n), j}) - f(\bar{X}_T^{0, (n), j}) ) + E[f(X_T^0)]
\]
to the multi-level Monte Carlo method via 
\[
\bar{P}_l^{\mathrm{new}} := f(\bar{X}_T^{\epsilon, (n_l)}) - f(\bar{X}_T^{0, (n_l)}) + E[f(X_T^0)].
\]

\begin{rem}
We can also consider another MLMC sampling method via 
\[
\bar{P}_l^{\mathrm{another}} := f(\bar{X}_T^{\epsilon, (n_l)} - \bar{X}_T^{0, (n_l)} + X_T^0), 
\]
whose computational cost is $O(\epsilon^2\gamma^{-2}(\log \gamma /\epsilon)^2)$ for Lipschitz functions $f$.
However, in this case we cannot take advantage of the explicit formula for the term $E[f(X_T^0)]$. 
\end{rem}

Giles \cite{G08} assumed that $f$ is Lipschitz continuous to analyze the variance of estimator. 
On the other hand, we need $f \in C_b^2 := 
\{ g \in C^2(\R^N; \R) \ | \ \partial_i g \mbox{ and } \partial_{ij}g \mbox{ are bounded, } 1 \leq i, j \leq N \}$ 
in order to use the asymptotics with respect to $\epsilon$  
(Notice that $|f(x_1)-f(y_1)+f(y_2)-f(x_2)| \not\leq C |x_1 -y_1 + y_2 - x_2|$ in general). 
Our analysis follows from the lemma below.

\begin{lem}\label{lemLIP}
For $f \in C_b^2$, 
\begin{align*}
|f(x_1)-f(y_1)+f(y_2)-f(x_2)| \leq&
\frac{\| \nabla^2 f\|_\infty}{2}(|x_1-x_2| + |y_1-y_2|)|x_1-y_1| 
\\ & + \| \nabla f\|_\infty |x_1 -y_1 + y_2 - x_2|.
\end{align*}
\end{lem}
\begin{proof}
This can be proved immediately by using the mean value theorem twice 
(See also Lemma \ref{lemB3}).
\end{proof}
Then we have the following variance estimate.
\begin{prop}\label{MLMCsmooth} Assume that $(H_1)$-$(H_4)$ hold. 
For $f \in C_b^2$, we have 
\begin{align*}
\mathrm{Var}^{1/2}(\bar{P}_l^{\mathrm{new}} - \bar{P}_{l-1}^{\mathrm{new}}) \leq C \epsilon n_{l}^{-1/2}.
\end{align*}
\end{prop}
\begin{proof}
The variance of the difference $\bar{P}_l^{\mathrm{new}} - \bar{P}_{l-1}^{\mathrm{new}}$ is estimated as 
\begin{align*}
\mathrm{Var}^{1/2}(\bar{P}_l^{\mathrm{new}} - \bar{P}_{l-1}^{\mathrm{new}}) 
\leq& \ \| f(\bar{X}_T^{\epsilon, (n_l)}) - f(\bar{X}_T^{0, (n_l)} )
-( f(\bar{X}_T^{\epsilon, (n_{l-1})}) - f(\bar{X}_T^{0, (n_{l-1})}))\|_2
\\ \leq& \ \| f(X_T^\epsilon) - f(\bar{X}_T^{\epsilon, (n_l)}) + f(\bar{X}_T^{0, (n_l)}) - f(X_T^0) \|_2
\\ & + \| f(X_T^\epsilon) - f(\bar{X}_T^{\epsilon, (n_{l-1})}) + f(\bar{X}_T^{0, (n_{l-1})}) - f(X_T^0) \|_2.
\end{align*}
Thus by Lemma \ref{lemLIP}, 
\begin{align*}
& \| f(X_T^\epsilon) - f(\bar{X}_T^{\epsilon, (n_l)}) + f(\bar{X}_T^{0, (n_l)}) - f(X_T^0) \|_2
\\ & \leq 
C ((\| X_T^\epsilon - X_T^0 \|_4 + \| \bar{X}_T^{\epsilon, (n_l)} - \bar{X}_T^{0, (n_l)}\|_4 )\|X_T^\epsilon - \bar{X}_T^{\epsilon, (n_l)}\|_4
\\ & \ \ \ \ + \| X_T^\epsilon - \bar{X}_T^{\epsilon, (n_l)} + \bar{X}_T^{0, (n_l)} - X_T^0\|_2 ).
\end{align*}
Hence using Lemma \ref{lemB1}-\ref{lemB2} and Theorem \ref{firsttheorem}, we get the rate of convergence.
\end{proof}

For the use of the multi-level Monte Carlo method, we have obtained the results as follows. 
\begin{align*}
E[\bar{P}_l^{\mathrm{new}} - P] &= O(\epsilon/n_l), 
\\ \mathrm{Var}^{1/2}(\bar{P}_l^{\mathrm{new}} - \bar{P}_{l-1}^{\mathrm{new}}) &= O(\epsilon/ n_l^{1/2}). 
\end{align*}
So the estimator for $(\bar{P}_l^{\mathrm{new}})_{l\geq 0}$ has an equivalent effect to the one for $(\bar{P}_l)_{l\geq 0}$ 
with the required error $O(\gamma/\epsilon)$.
Consequently we get the order of computational cost $O(\epsilon^2 \gamma^{-2} (\log(\gamma / \epsilon))^{-2})$. 
Both the asymptotic method and multi-level Monte Carlo method are very easily computable, 
so that practitioners will get large benefit only with small additional implementation cost. 

\begin{rem}
Clearly, we can also check the variance estimate for the accelerated Milstein scheme. 
Let $\hat{P}_l^{\mathrm{new}} := f(\hat{X}_T^{\epsilon, (n_l)}) - f(\hat{X}_T^{0, (n_l)}) + E[f(X_T^0)]$. 
By a similar argument, we derive that under $(H_1')$-$(H_4')$ and $f \in C_b^2$, 
\begin{align}\label{M_variance}
\mathrm{Var}^{1/2}(\hat{P}_l^{\mathrm{new}} - \hat{P}_{l-1}^{\mathrm{new}}) \leq C \epsilon n_{l}^{-1}. 
\end{align} 
We have not obtained weak convergence results for the accelerated Milstein scheme yet. 
However, we guess that from the basic proof of Takahashi-Yoshida \cite{TY05}, it holds that 
\begin{align}\label{M_weak}
E[f(X^\epsilon_T)] - ( E[f(\hat{X}_T^{\epsilon, (n)})] - E[f(\hat{X}_T^{0, (n)})] + E[f(X^0_T)] ) = O\Big(\frac{\epsilon}{n}\Big)
\end{align}
under some smoothness conditions for $f$ and the coefficients of $X_t^\epsilon$. 
Thus combining the results (\ref{M_variance}), (\ref{M_weak}) and the discussion in Giles \cite{G08, G07}, 
we finally conclude that the total computational cost is $O(\epsilon^2\gamma^{-2})$.
\end{rem}

\subsection{Lipschitz payoffs}
Let us consider the first component $(X_T^\epsilon)^{(1)}$ as an asset dynamics. 
Our interest is pricing an option $f((X_T^\epsilon)^{(1)})$ with Lipschitz payoffs 
$f:\R \rightarrow \R$. 
Set $\bar{P}_l^{\mathrm{new}} = f((\bar{X}_T^{\epsilon, (n_l)})^{(1)}) - f((\bar{X}_T^{0, (n_l)})^{(1)}) + E[f((X_T^0)^{(1)})].$ 
Then we can obtain an upper bound estimate as follows. 

\begin{thm}\label{thm:mlmc}
Assume $(H_1)$-$(H_4)$ and 
$f: \R \rightarrow \R$ is a Lipschitz continuous function 
whose weak derivative has bounded variation in $\R$. 
In addition, suppose $(X_T^0)^{(1)}$ has a bounded density, and  
$(\bar{X}_T^0)^{(1)}$ also has a bounded density uniformly with respect to $n$. 
Then we have for any small $\delta >0$, 
\begin{align*}
\mathrm{Var}^{1/2}(\bar{P}_l^{\mathrm{new}} - \bar{P}_{l-1}^{\mathrm{new}}) \leq C \epsilon^{(1-\delta)/2} n_{l}^{-1/2}.
\end{align*}
\end{thm}
\begin{proof}
See \ref{sec:proofmlmc}.
\end{proof}

This theorem implies that the required computational cost turns out to be 
$O(\epsilon^{1-\delta} \gamma^{-2}(\log \gamma/\epsilon)^{2})$, 
with $L = \log(\epsilon \gamma^{-1}) / \log(k) + O(1)$ and $N_l = O(\epsilon^{1-\delta} \gamma^{-2} L n_l^{-1})$.

We now summarize strong rate of convergence 
for $\bar{P}_l - \bar{P}_{l-1}$ and $\bar{P}_l^{\mathrm{new}} - \bar{P}_{l-1}^{\mathrm{new}}$ 
in Table \ref{tabular_order}. 

\begin{table}[htb]
\begin{center}
\begin{tabular}{|c|c|c|}
\hline
Payoff & Standard E-M & Accelerated E-M \\
\hline\hline
$C^2_b$ & $O(n_l^{-1/2})$ & $O(\epsilon n_l^{-1/2})$ \\
\hline
Lipschitz & $O(n_l^{-1/2})$ & $O(\epsilon^{(1-\delta)/2} n_l^{-1/2})$ \\
\hline
Digital & $O(n_l^{-(1-\delta)/4})$, (\cite{A09}, \cite{GHM09}) & - \\
\hline
\end{tabular}
\caption{Strong rate of convergence of $\bar{P}_l - \bar{P}_{l-1}$ and $\bar{P}_l^{\mathrm{new}} - \bar{P}_{l-1}^{\mathrm{new}}$.}
\label{tabular_order}
\end{center}
\end{table}

\subsection{Localization for irregular payoffs}

The regularity of $f$ seems to be essential for the accelerated MLMC method introduced in previous. 
For example, we will see through computational experiments that the acceleration with discontinuous functions $f$ does not work so well.

We now propose a localization technique for this problem. 
Let us define a decomposition 
\[
f = f_s + f_{ir}
\]
where $f_s$ is a smooth (at least Lipschitz continuous) function with $f \approx f_s$. 
Then we apply the accelerated MLMC to the smooth part $f_s$ 
and the standard MLMC to the irregular part $f_{ir}$. 
In other words, we consider the MLMC method for 
\[
\bar{P}_l^{\mathrm{new}}(f_s) := f(\bar{X}_T^{\epsilon}) - f_s(\bar{X}_T^{0}) + E[f_s(X_T^{0})].
\]
The standard MLMC for discontinuous functions was studied in Avikainen \cite{A09} and Giles et al.\ \cite{GHM09}.

\section{Simulations}

\subsection{Numerical experiments for SABR model}
In this section, we want to study an estimator of $L^2$-norm 
$(\frac{1}{M} \sum_{j=1}^M ( X_T^{\epsilon, j} - \tilde{Y}_T^{\epsilon, (n), j} )^2 )^{1/2}$ for the SABR model. 
As a reference path, we use $\bar{X}_T^{\epsilon, (n_{\mathrm{ref}})}$ instead of $X_T^{\epsilon}$.  

We set the parameters as follows. 
\begin{itemize}
\item $S_0 = 100$, $\beta = 0.9$, $\alpha_0 = 0.16\times S_0^{2(1-\beta)}$, $\nu = 0.1$, $\rho = -0.6$, $T=1$
\item $n_{\mathrm{ref}} = 2^{14}$, $n=8, 16, 32, 64, 128, 256$. 
\end{itemize}
Here we considered a scaling for $\alpha_0$ (via $S_t \approx S_0^{1-\beta}S_t^\beta$).
The number of simulation $M$ for the test is $10^5$. 
The results are given in Figure \ref{fig2}. 
The accelerated scheme is faster than the standard method in both cases of $L^2$-error.

\begin{figure}[htb]
  \begin{center}
    \begin{tabular}{cc}
      \includegraphics[width=60mm]{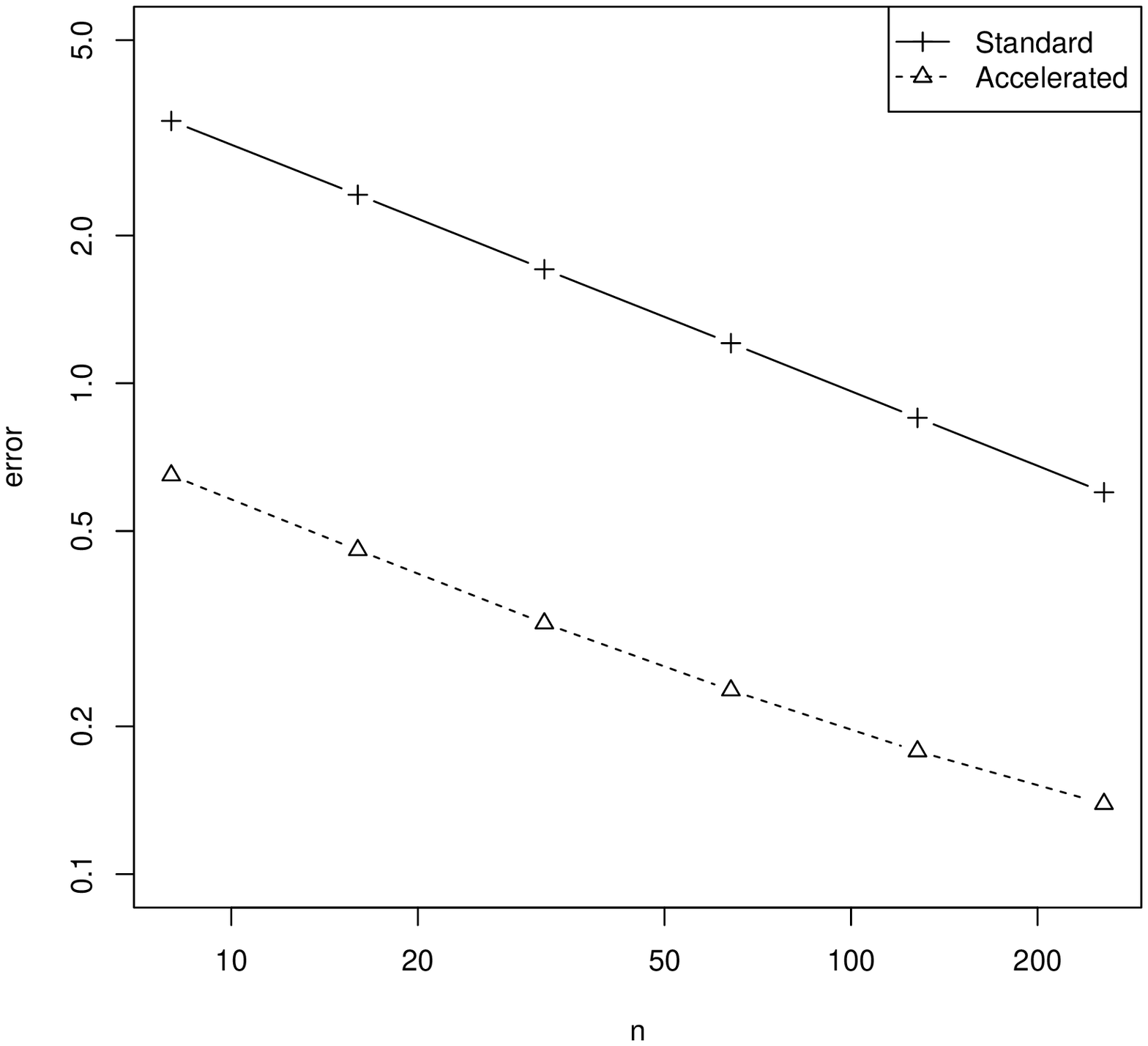} &
      \includegraphics[width=60mm]{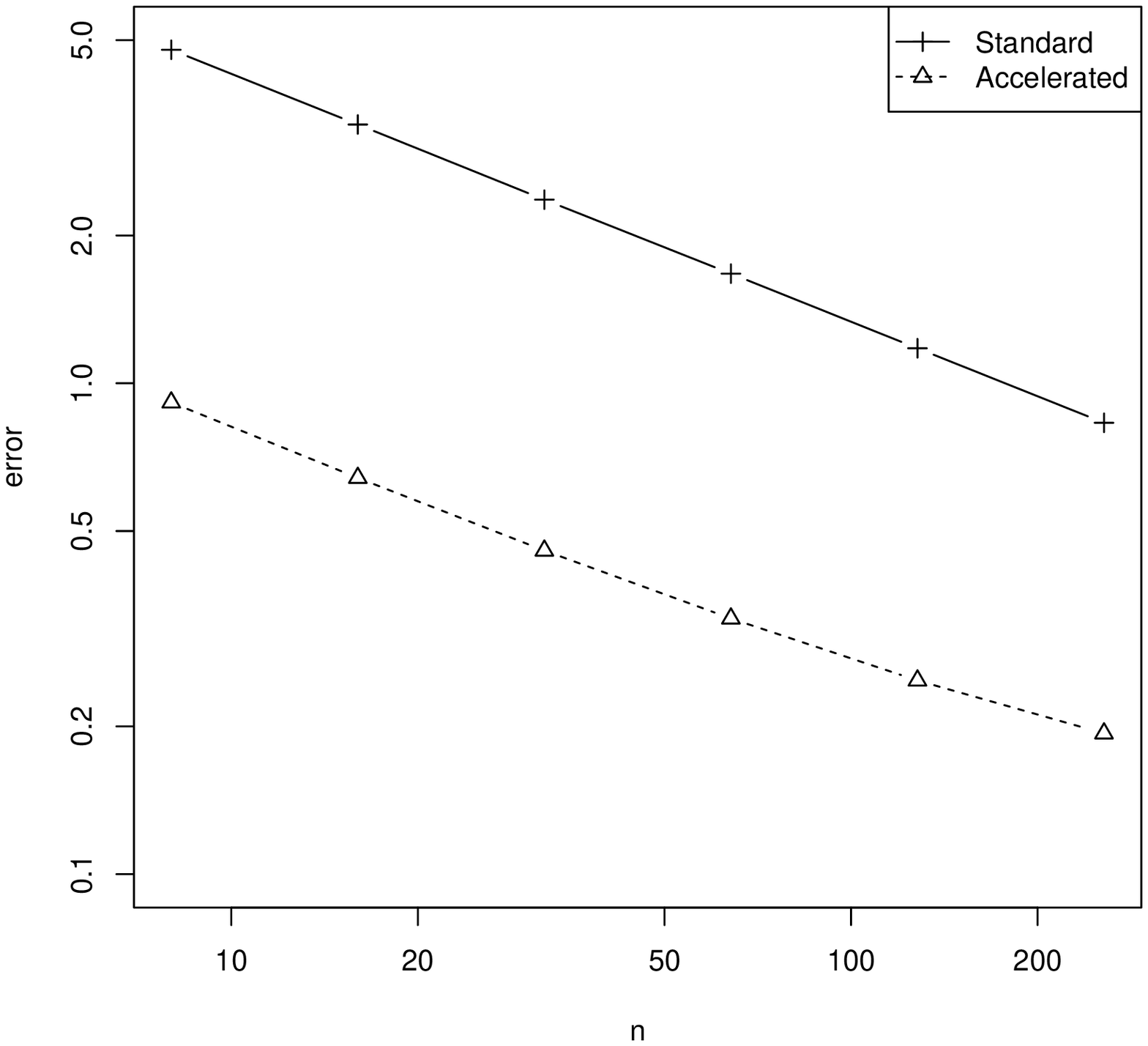} 
    \end{tabular}
    \caption{$L^2$-error : $E[|\bar{S}_T^{(n_{\mathrm{ref}})} - \tilde{Y}_T^{(n)}|^2]^{1/2}$ for the left and 
     $E[\max_{0\leq i \leq n_{\mathrm{ref}}}|\bar{S}_{t_i}^{(n_{\mathrm{ref}})} - \tilde{Y}_{t_i}^{(n)}|^2]^{1/2}$ for the right.}
    \label{fig2}
  \end{center}
\end{figure}

We next study the case with several $\nu$. 
Let us compute the $L^2$-error ratio for a random variable $Z$ which is defined as 
\[
\frac{E[|\bar{S}_T^{(n_{\mathrm{ref}})} - Z|^2]^{1/2}}{E[|\bar{S}_T^{(n_{\mathrm{ref}})} - \bar{S}_T^{(n)}|^2]^{1/2}} 
\times 100 \mbox{ (\%). }
\]
We fix the other parameters in the previous. 
In Figure \ref{fig3}, we can check the efficiency of the asymptotic method 
(only) when $\nu$ is small enough.

Finally we compare $\tilde{Y}$ and $\check{Y}$ with different $\beta$. 
Figure \ref{fig4} shows that the efficiency of $\check{Y}$ is very close to that of $\tilde{Y}$ as $\beta \approx 1$. 
Therefore if $\beta \approx 1$, we can apply the analytical tractability of $\check{Y}$ 
to pathwise simulation, computing expectations, or so on.

\begin{figure}[htbp]
  \begin{center}
    \includegraphics[width=75mm]{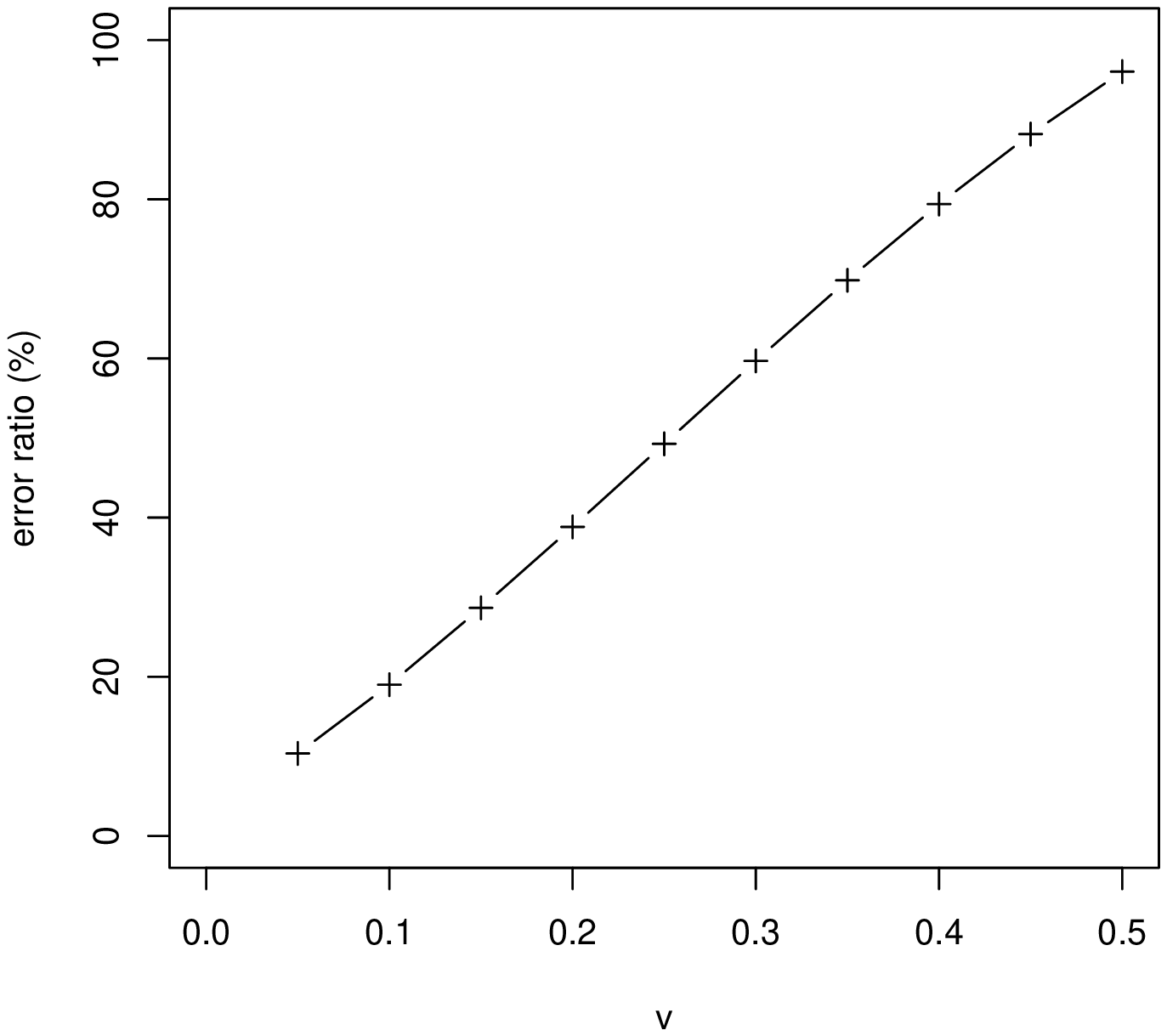} 
    \caption{$L^2$-error ratio for $\tilde{Y}$ with different $\nu$.}
    \label{fig3}
  \end{center}
  \begin{center}
    \includegraphics[width=75mm]{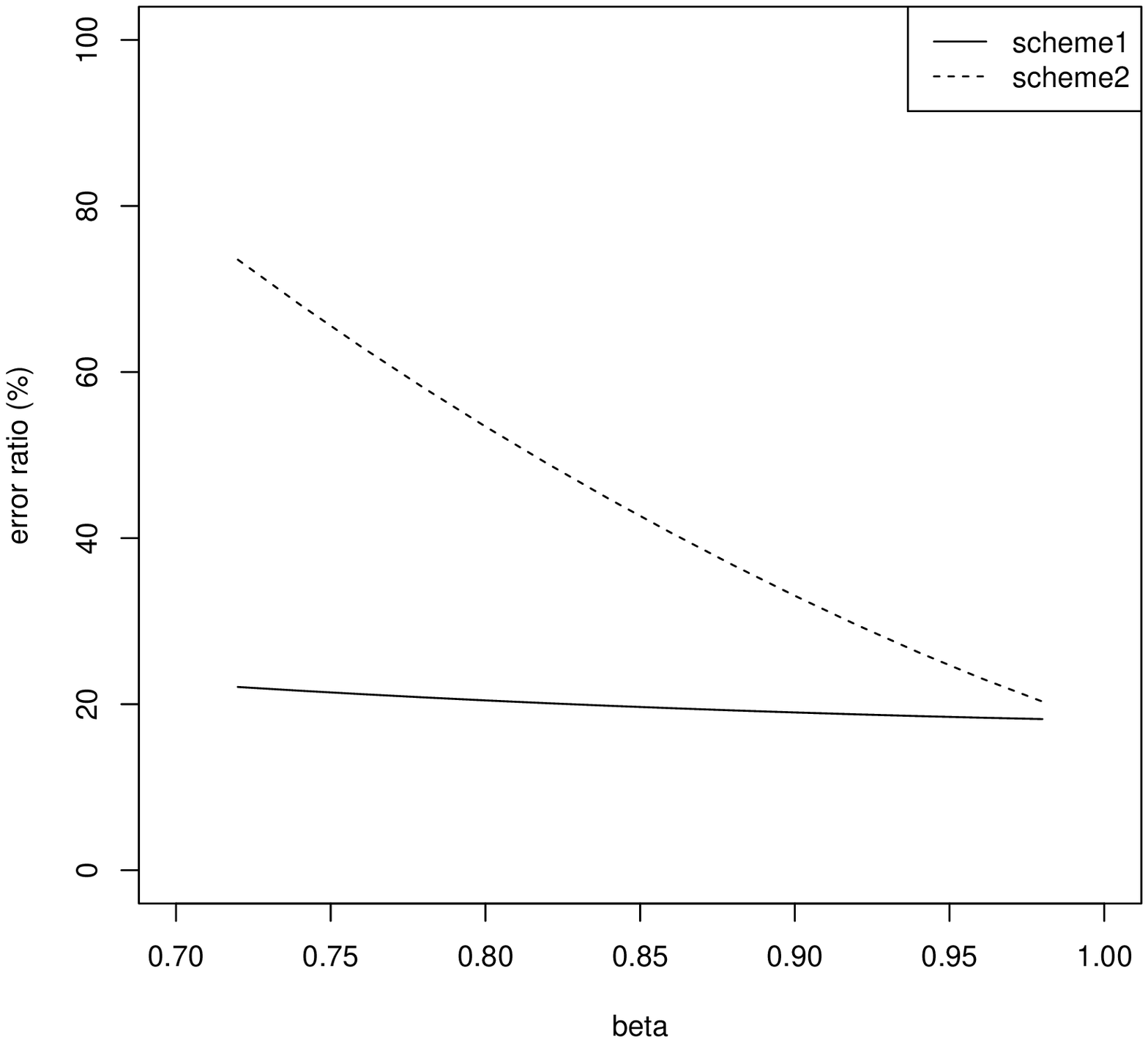} 
    \caption{$L^2$-error ratio for $\tilde{Y}$ (scheme$1$) and $\check{Y}$ (scheme$2$) with different $\beta$.}
    \label{fig4}
  \end{center}
\end{figure}

\subsection{Numerical tests for MLMC}
To show that the accelerated method is more efficient than the standard method with MLMC,  
we take a numerical test for 
$E[\bar{P}_l^{\mathrm{new}} - P] = O(\epsilon/n_l)$, 
and $\mathrm{Var}^{1/2}(\bar{P}_l^{\mathrm{new}} - \bar{P}_{l-1}^{\mathrm{new}}) = O(\epsilon/ n_l^{1/2})$ 
under the SABR model with small parameter $\nu$. 
Let us consider payoff functions (European and digital options)
$$f(x) = \max(0, x-100) \mbox{ or } f(x) = 1_{\{x-100 \geq 0\}}$$ 
and the parameters 
\begin{itemize}
\item $S_0 = 100$, $\beta = 1$, $\alpha_0 = 0.16$, $\nu = 0.1$, $\rho = -0.6$, $T=1$
\end{itemize}
The level structure of MLMC is given by $k=4$, i.e., $n_l = 4^l$. As a localization for digital option, we use 
\[
f_s(x) = (\max(x-100+h,0) - \max(x-100-h,0))/2h.
\]
Here we set $h=1.0$.

Figure \ref{figMLMC} and \ref{figMLMC_dig} show the numerical results. We used the number of simulation $M=10^7$ for the left, and $M=10^5$ for the right. 
The results basically imply that the accelerated method works better than the standard one as in preceding numerical experiments. 
Remarkably the accelerated method performs worse in the case of variance estimates for digital option, likely due to discontinuity of the payoff function.
In contrast, the localized scheme (Accelerated\_loc) performs better than the others to some extent.
We note that for general $1/2 \leq \beta <1$, the (semi-)analytical formula for CEV option pricing model 
can be used in order to compute $E[f(S_T^0)]$ (See \cite{Sch89}). 

\begin{figure}[htbp]
  \begin{center}
     \begin{tabular}{cc}
     \includegraphics[width=60mm]{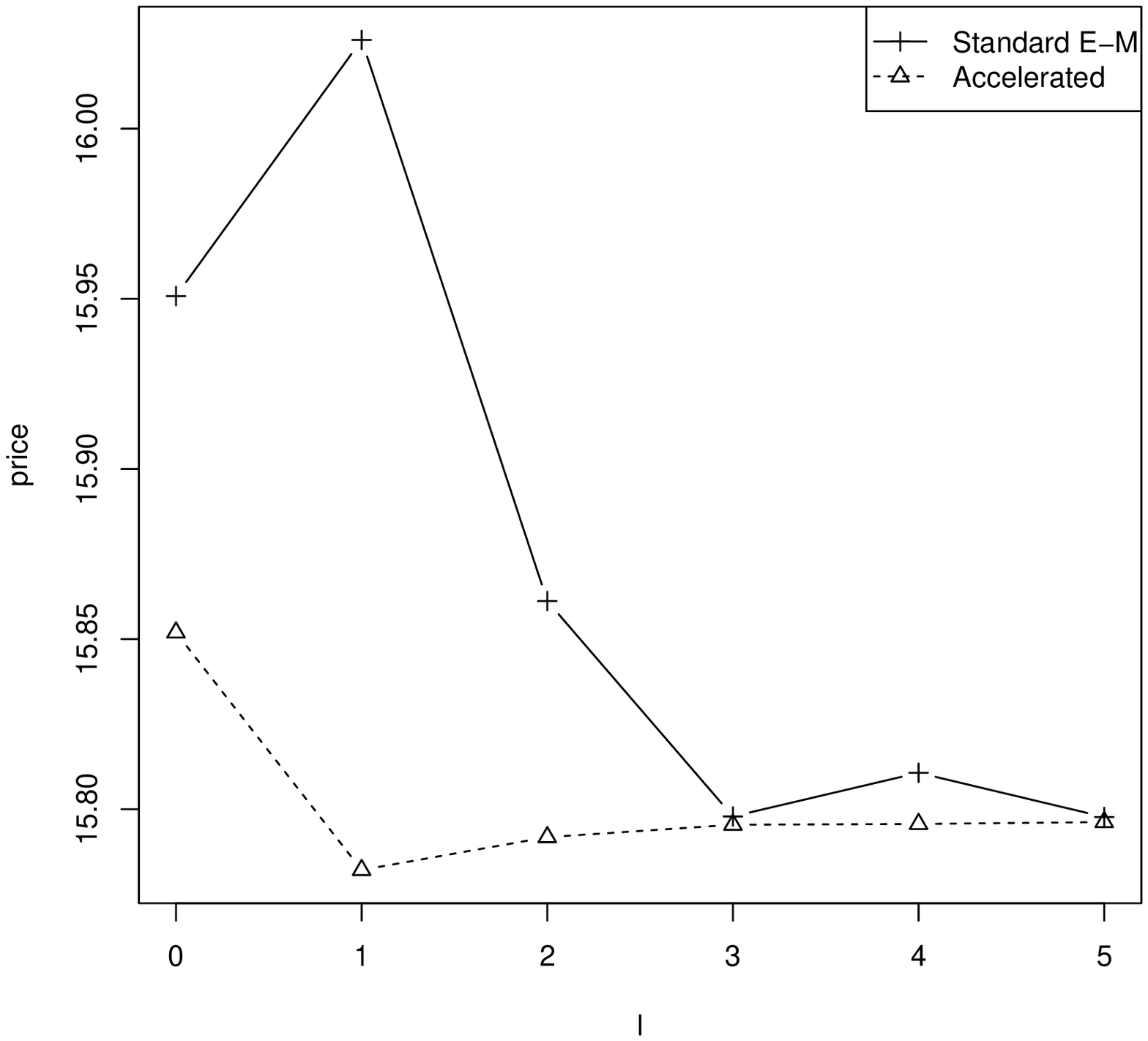}
     \includegraphics[width=60mm]{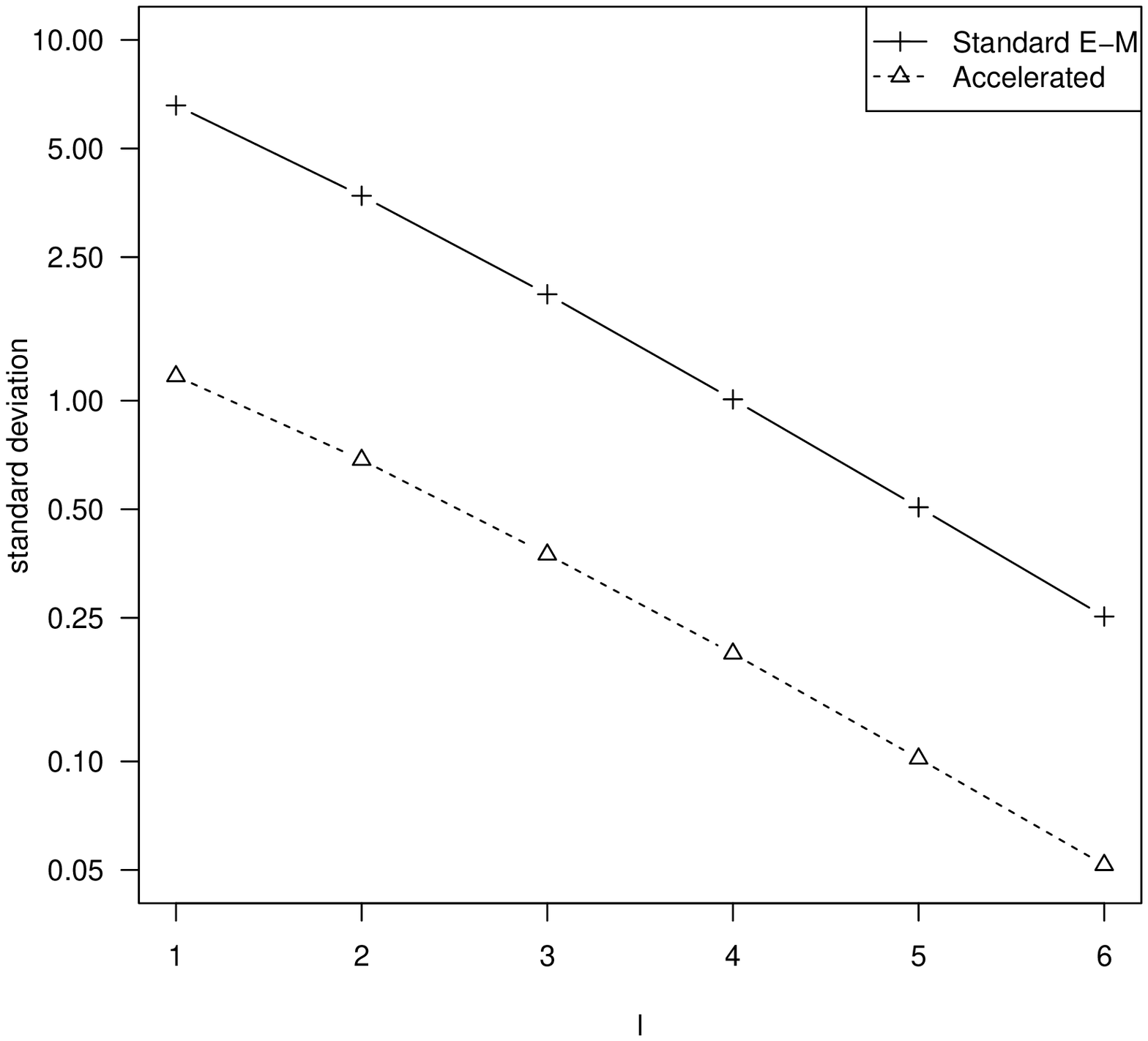}
     \end{tabular}
    \caption{European option: (Left) A comparison of weak convergence between $\bar{P}_l$ and $\bar{P}_l^{\mathrm{new}}$. 
    (Right) A comparison of standard deviation between $\bar{P}_l-\bar{P}_{l-1}$ and $\bar{P}_l^{\mathrm{new}} - \bar{P}_{l-1}^{\mathrm{new}}$. }
    \label{figMLMC}
  
  \vspace{5mm}
  
    \begin{tabular}{cc}
     \includegraphics[width=60mm]{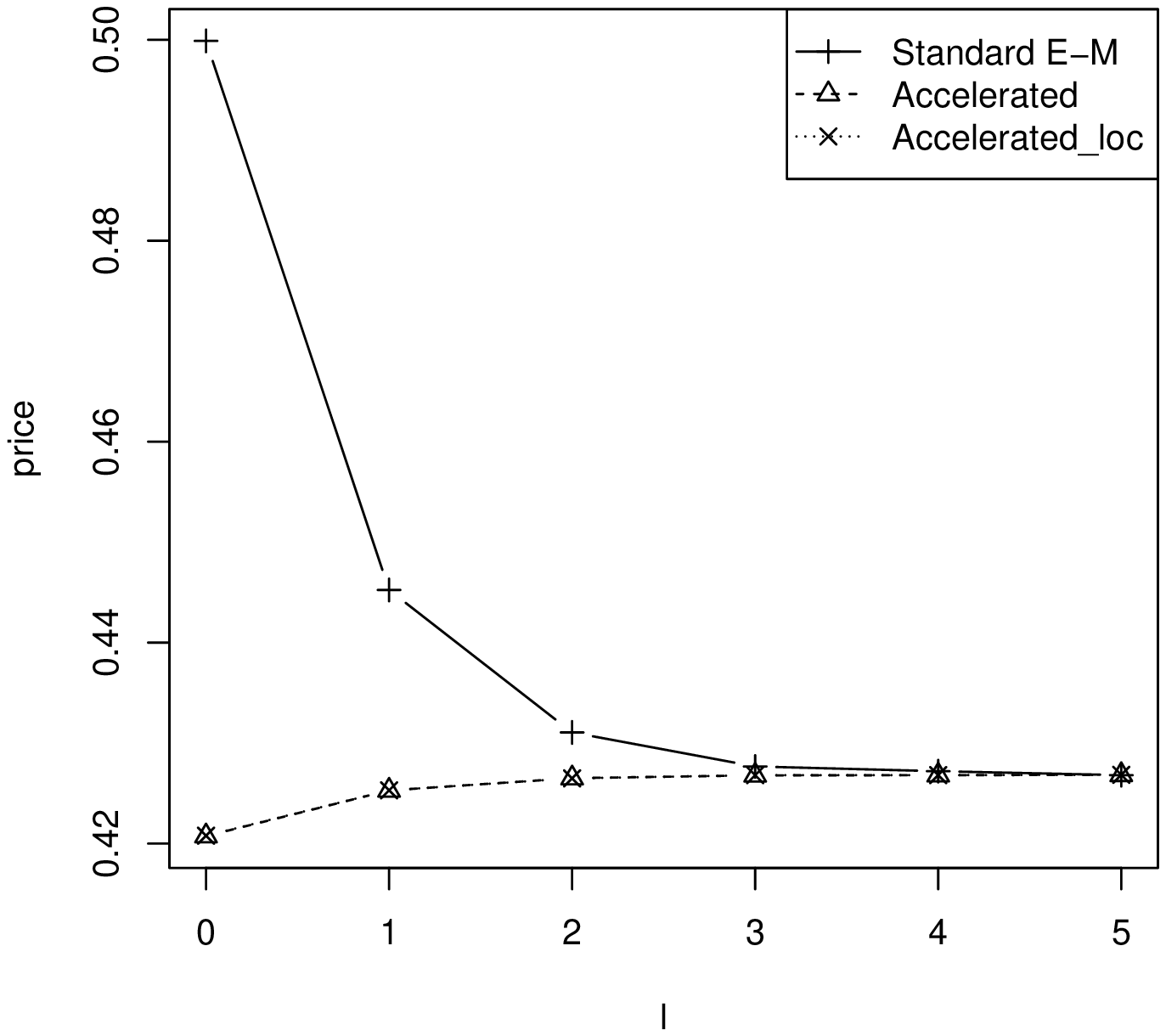}
     \includegraphics[width=60mm]{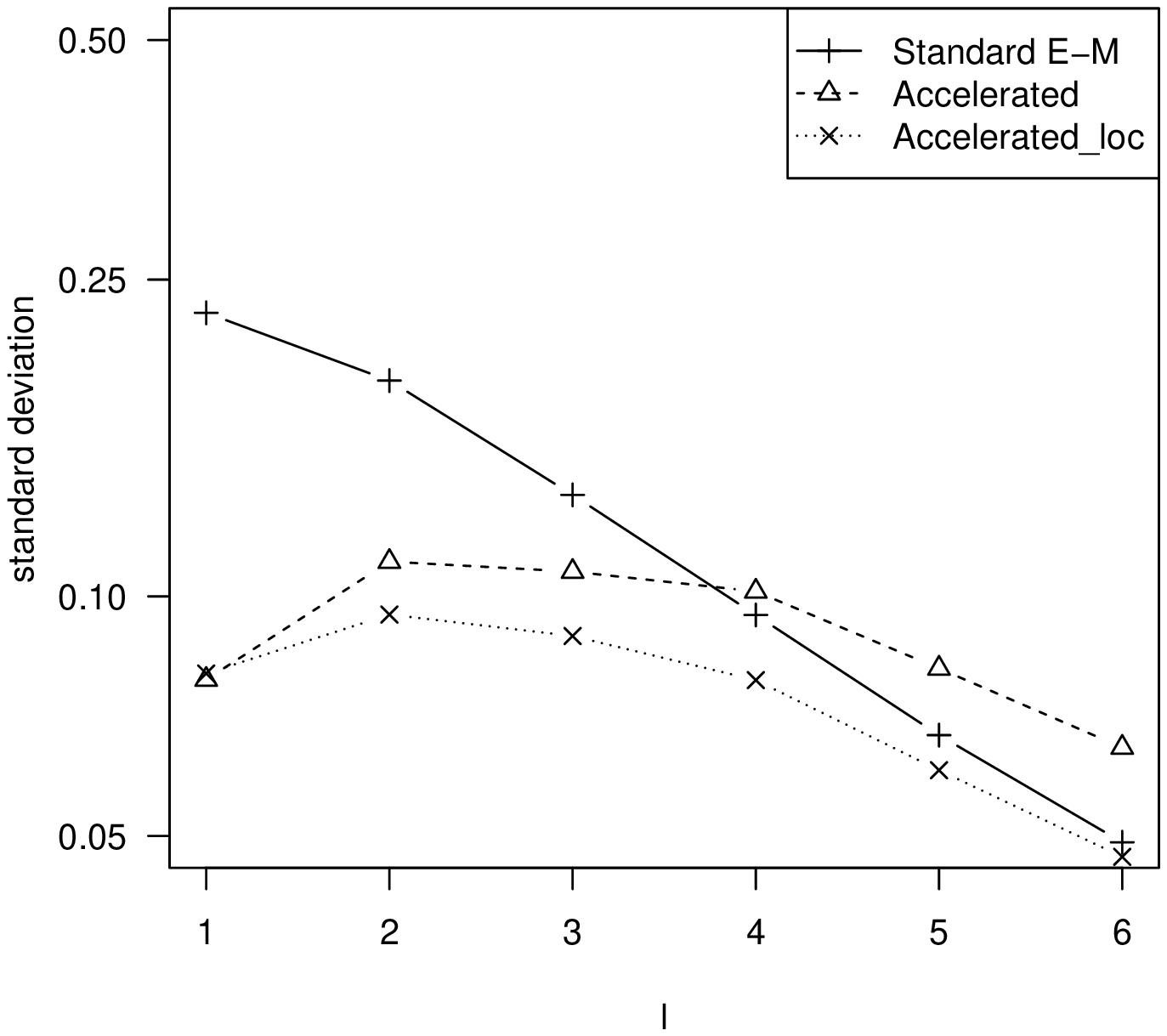}
     \end{tabular}
    \caption{Digital option: (Left) A comparison of weak convergence between $\bar{P}_l$ and $\bar{P}_l^{\mathrm{new}}$. 
    (Right) A comparison of standard deviation between $\bar{P}_l-\bar{P}_{l-1}$ and $\bar{P}_l^{\mathrm{new}} - \bar{P}_{l-1}^{\mathrm{new}}$. }
    \label{figMLMC_dig}
  \end{center}
\end{figure}

\section{Discussion: small extension}
By the accelerated method, the Euler-Maruyama scheme can be faster 
with a parameter $\epsilon$ small enough. 
But of course, if $\epsilon$ is not so small, 
the method does not work effectively (See Figure \ref{fig3}).

To improve the efficiency of the method, we consider a natural extended scheme 
$$\bar{F}^\epsilon - \rho_\epsilon(\bar{F}^0 - F^0)$$
where $\rho_\epsilon = 1 + O(\epsilon)$ to keep the rate of convergence. 
For example, the optimal $\rho_\epsilon$ with respect to $L^2$-error is given by 
\begin{align*}
\rho_\epsilon(L^2) &:= \arg\min_\rho E[(F^\epsilon - \bar{F}^\epsilon + \rho(\bar{F}^0 - F^0) )^2]
\\ &= \frac{E[(F^\epsilon - \bar{F}^\epsilon)(F^0 - \bar{F}^0)]}{E[(\bar{F}^0 - F^0)^2]}
\end{align*}
However, since the value of $F^\epsilon$ is unknown, 
we can not estimate $\rho_\epsilon(L^2)$ directly. 
In general it does not seem easy to choose an appropriate and computable $\rho_\epsilon$. 
This issue is left for future work.

\appendix

\newtheorem{Alem}{Lemma}[section]
\renewcommand{\theAlem}{\Alph{section}\arabic{Alem}}

\section{Proof of Theorem \ref{firsttheorem} and \ref{secondtheorem}}

We use the following notations.
\begin{itemize}
\item $\eta(s) := t_i$ if $s \in [t_i, t_{i+1})$.
\item $\bar{X}_t^{\epsilon} \equiv \bar{X}_t^{\epsilon, (n)}$, $\hat{X}_t^{\epsilon} \equiv \hat{X}_t^{\epsilon, (n)}$.
\end{itemize}
We will apply the Burkholder-Davis-Gundy (BDG) inequality 
\[
c_p E[\langle M \rangle_T^{p/2}] \leq E[\sup_{0\leq t \leq T}|M_t|^p] \leq C_p E[\langle M \rangle_T^{p/2}]
\]
to the proofs below:
Here $p >0$ and $M_t$ is a continuous local martingale.

Using the BDG inequality and Gronwall inequality, 
we can show the following moment estimates. 
(See \cite{KP92} for the proof in the case of $L^2$-norm.)
\begin{Alem}\label{lemB1}
(i) Suppose that the assumptions $(H_1)$-$(H_2)$ hold. Then for any $p \geq 2$, we have
\begin{align*}
\sup_{\epsilon \in [0,1]} E[\sup_{0\leq t \leq T}|X_t^\epsilon|^p] 
+ \sup_{\epsilon \in [0,1]} E[\sup_{0\leq t \leq T}|\bar{X}_t^\epsilon|^p] 
< \infty, 
\\ 
\sup_{\epsilon \in [0,1]} \max_{0 \leq i \leq n-1}E[\sup_{t_i\leq t \leq t_{i+1}}|X_t^\epsilon - X_{t_i}^\epsilon|^p] \leq C(T, x_0, p) / n^{1/2}. 
\\ 
\sup_{\epsilon \in [0,1]} E[\sup_{0\leq t \leq T}|X_t^\epsilon - \bar{X}_t^\epsilon|^p] \leq C(T, x_0, p) / n^{1/2}. 
\end{align*}

(ii) Suppose that the assumptions $(H'_1)$-$(H'_2)$ hold. Then for any $p \geq 2$, we have
\begin{align*}
\sup_{\epsilon \in [0,1]} E[\sup_{0\leq t \leq T}|\hat{X}_t^\epsilon|^p] < \infty, 
\\
\sup_{\epsilon \in [0,1]} E[\sup_{0\leq t \leq T}|X_t^\epsilon - \hat{X}_t^\epsilon|^p] \leq C(T, x_0, p) / n. 
\end{align*}
\end{Alem}

We now give an important lemma for the proof of the main theorems.
\begin{Alem}\label{lemB2}
(i) Under $(H_1)$-$(H_3)$, we have for any $p \geq 2$, 
\begin{align}\label{B2_1}
E\Big[\sup_{0\leq t \leq T}| X_{t}^\epsilon - X_t^{0} |^p \Big]^{1/p} 
\leq C(T, x_0, p) \epsilon,
\end{align}
and 
\begin{align}\label{B2_2}
\max_{0 \leq i \leq n-1}E\Big[\sup_{t_i\leq t \leq t_{i+1}}| X_{t}^\epsilon - ( X_{t_i}^{\epsilon} - X_{t_i}^{0} + X_t^{0}) |^p \Big]^{1/p} 
\leq C(T, x_0, p) \frac{\epsilon}{n^{1/2}}.
\end{align}

(ii) Under $(H_1)$-$(H_3)$, for any $p \geq 2$,
\begin{align*}
E\Big[\sup_{0\leq t \leq T}| \bar{X}_{t}^\epsilon - \bar{X}_t^{0} |^p \Big]^{1/p} 
\leq C(T, x_0, p) \epsilon.
\end{align*}

(iii) Under $(H'_1)$-$(H'_3)$, for any $p \geq 2$,
\begin{align*}
E\Big[\sup_{0\leq t \leq T}| \hat{X}_{t}^\epsilon - \hat{X}_t^{0} |^p \Big]^{1/p} 
\leq C(T, x_0, p) \epsilon.
\end{align*}
\end{Alem}

\begin{proof}
(i): 
We first note that
\[
X_t^\epsilon - X_t^0 = 
\int_0^t(b(X_s^\epsilon, \epsilon) - b(X_s^0, 0))ds
+ \int_0^t(\sigma(X_s^\epsilon, \epsilon) - \sigma(X_s^0, 0))dBs, 
\]
and by the BDG inequality for the stochastic integral term, 
\begin{align*}
E[\sup_{0\leq s \leq t}|X_s^\epsilon - X_s^0|^p] \leq C_p & 
\Big(
E\Big[\big(\int_0^t|b(X_s^\epsilon, \epsilon) - b(X_s^0, 0)|ds\big)^p\Big] 
\\ & + E\Big[ \big(\int_0^t(\sigma(X_s^\epsilon, \epsilon) - \sigma(X_s^0, 0))^2ds \big)^{p/2} \Big] 
\Big).
\end{align*}
Using the conditions $(H_1)$-$(H_3)$ for the above, we have immediately 
\begin{align*}
G(t) := E[\sup_{0\leq s \leq t}| X_s^\epsilon - X_s^0 |^p] 
\leq C_1 \epsilon^p + C_2 \int_0^t G(s) ds.
\end{align*}
Here the constants $C_1$ and $C_2$ do not depend on $\epsilon$. 
Thus from the Gronwall inequality we obtain (\ref{B2_1}).

We next consider the second result (\ref{B2_2}). 
Since
\[
X_{t}^\epsilon - ( X_{\eta(t)}^{\epsilon} - X_{\eta(t)}^{0} + X_t^{0}) 
= \int_{\eta(t)}^t (b(X_s^\epsilon, \epsilon) - b(X_s^0, 0)) ds
+ \int_{\eta(t)}^t (\sigma(X_s^\epsilon, \epsilon) - \sigma(X_s^0, 0)) dB_s, 
\]
the inequality (\ref{B2_2}) follows from $(H_2)$-$(H_3)$ and (\ref{B2_1}). 

The proofs for (ii) and (iii) are straightforward as in (\ref{B2_1}). 
\end{proof}

The following lemma will be used such as the Lipschitz continuous property.
\begin{Alem}\label{lemB3}
(i) 
Assume that $(H_1)$-$(H_4)$ hold. Then 
\begin{align*}
&|b(x_1, \epsilon) - b(y_1, \epsilon) + b(y_2, 0) - b(x_2, 0)| 
\\ &\leq C ((\epsilon + |x_1-x_2| + |y_1-y_2|)(x_1-y_1) + |x_1 - y_1 + y_2 - x_2|).
\\ &|\sigma(x_1, \epsilon) - \sigma(y_1, \epsilon) + \sigma(y_2, 0) - \sigma(x_2, 0)| 
\\ &\leq C ((\epsilon + |x_1-x_2| + |y_1-y_2|)(x_1-y_1) + |x_1 - y_1 + y_2 - x_2|).
\end{align*}

(ii) 
Assume that $(H'_1)$-$(H'_4)$ hold. Then 
\begin{align*}
& |\sigma\sigma'(x_1, \epsilon) - \sigma\sigma'(y_1, \epsilon) + \sigma\sigma'(y_2, 0) - \sigma\sigma'(x_2, 0)| 
\\ & \leq C ((\epsilon + |x_1-x_2| + |y_1-y_2|)(x_1-y_1) + |x_1 - y_1 + y_2 - x_2|).
\end{align*}
\end{Alem}

\begin{proof}
We only prove for $b$. By the mean value theorem, 
\begin{align*}
b(x_1, \epsilon) - b(y_1, \epsilon) + b(y_2, 0) - b(x_2, 0) = \xi_{x_1, y_1}^\epsilon(x_1-y_1) + \xi_{y_2, x_2}^0(y_2-x_2)
\end{align*}
where $\xi_{x, y}^\epsilon := \int_0^1 \partial b (\rho x + (1-\rho) y , \epsilon)d\rho = \xi_{y, x}^\epsilon$. 
Taking the difference again in the right hand side, we have
\begin{align*}
\xi_{x_1, y_1}^\epsilon(x_1-y_1) = (\xi_{x_1, y_1}^\epsilon - \xi_{x_2, y_2}^0)(x_1-y_1) + \xi_{y_2, x_2}^0(x_1-y_1).
\end{align*}
Finally, using the assumption $(H_4)$ for $(\xi_{x_1, y_1}^\epsilon - \xi_{x_2, y_2}^0)$, we obtain the result.
\end{proof}

Now we shall prove the theorems. 

\begin{proof}[Proof of Theorem \ref{firsttheorem}]
Let us define 
\begin{align*}
G_1(t) := E[\sup_{0\leq s \leq t} |X_s^{\epsilon} - \bar{Y}_s^{\epsilon, (n)}|^p].
\end{align*}
By using the Gronwall inequality, our goal becomes to show the following:
\begin{align*}
G_1(t) \leq C_1\frac{\epsilon^p}{n^{p/2}} + C_2 \int_0^t G_1(s)ds
\end{align*}
where $C_1$ and $C_2$ depend only on $T, x_0, p$.

We now compute 
\[
X_t^{\epsilon} - \bar{Y}_t^{\epsilon, (n)} = e^\epsilon(t) + \bar{e}^\epsilon(t)
\]
where 
\begin{align*}
e^\epsilon(t) =& \int_0^t (b(X_{\eta(s)}^\epsilon, \epsilon) - b(\bar{X}_{\eta(s)}^\epsilon, \epsilon) 
+ b(\bar{X}_{\eta(s)}^0, 0) - b(X_{\eta(s)}^0, 0)) ds 
\\ & + \int_0^t (\sigma(X_{\eta(s)}^\epsilon, \epsilon) - \sigma(\bar{X}_{\eta(s)}^\epsilon, \epsilon) 
+ \sigma(\bar{X}_{\eta(s)}^0, 0) - \sigma(X_{\eta(s)}^0, 0)) dB_s
\end{align*}
and
\begin{align*}
\bar{e}^\epsilon(t) =& \int_0^t (b(X_{s}^\epsilon, \epsilon) - b(X_{\eta(s)}^\epsilon, \epsilon)
+ b(X_{\eta(s)}^0, 0) - b(X_{s}^0, 0)) ds 
\\ & + \int_0^t (\sigma(X_{s}^\epsilon, \epsilon) - \sigma(X_{\eta(s)}^\epsilon, \epsilon)
+ \sigma(X_{\eta(s)}^0, 0) - \sigma(X_{s}^0, 0)) dB_s.
\end{align*}
For $\bar{e}^\epsilon(t)$, we obtain from Lemma \ref{lemB3}, 
\begin{align*}
& |b(X_{s}^\epsilon, \epsilon) - b(X_{\eta(s)}^\epsilon, \epsilon)
+ b(X_{\eta(s)}^0, 0) - b(X_{s}^0, 0)| 
\\ &\leq \ C( (\epsilon + |X_{s}^\epsilon - X_{s}^0| + |X_{\eta(s)}^\epsilon - X_{\eta(s)}^0|)|X_{s}^\epsilon - X_{\eta(s)}^\epsilon| 
\\ & \ \ \ \  + | X_{s}^\epsilon - ( X_{\eta(s)}^{\epsilon} - X_{\eta(s)}^{0} + X_s^{0}) |),
\end{align*}
and
\begin{align*}
&|\sigma(X_{s}^\epsilon, \epsilon) - \sigma(X_{\eta(s)}^\epsilon, \epsilon)
+ \sigma(X_{\eta(s)}^0, 0) - \sigma(X_{s}^0, 0)| 
\\ & \leq \ C( (\epsilon + |X_{s}^\epsilon - X_{s}^0| + |X_{\eta(s)}^\epsilon - X_{\eta(s)}^0|)|X_{s}^\epsilon - X_{\eta(s)}^\epsilon| 
\\ & \ \ \ \ + | X_{s}^\epsilon - ( X_{\eta(s)}^{\epsilon} - X_{\eta(s)}^{0} + X_s^{0}) |).
\end{align*}
Hence the integral term in $\bar{e}^\epsilon(t)$ is evaluated by 
\begin{align*}
& E[\sup_{0\leq r \leq t}|\int_0^r (b(X_{s}^\epsilon, \epsilon) - b(X_{\eta(s)}^\epsilon, \epsilon)
+ b(X_{\eta(s)}^0, 0) - b(X_{s}^0, 0))ds |^p]
\\ & \leq \ C_3
\Big( 
(\epsilon^p + 2 \sup_{0\leq s \leq t} \|X_{s}^\epsilon - X_{s}^0\|_{2p}^p ) 
\sup_{0\leq s \leq t} \|X_{s}^\epsilon - X_{\eta(s)}^\epsilon\|_{2p}^p 
\\ & \ \ \ \ + 
\sup_{0\leq s \leq t} \| X_{s}^\epsilon - ( X_{\eta(s)}^{\epsilon} - X_{\eta(s)}^{0} + X_s^{0}) \|_p^p
\Big)
\end{align*}
By using the BDG inequality, the stochastic integral term in $\bar{e}^\epsilon(t)$ also 
has the same bound (except the size of constant $C_3$). Consequently, we have by Lemma \ref{lemB1}, \ref{lemB2}, 
\[
E[\sup_{0\leq s \leq t} |\bar{e}^\epsilon(s)|^p] \leq C_4 \frac{\epsilon^p}{n^{p/2}}.
\]
Applying a similar calculus to $e^\epsilon(t)$, we also get 
\begin{align*}
E[\sup_{0\leq s \leq t} |e^\epsilon(s)|^p] &\leq C_5 \frac{\epsilon^p}{n^{p/2}} + C_6 \int_0^t 
E[| X_{\eta(s)}^\epsilon - (\bar{X}_{\eta(s)}^\epsilon - \bar{X}_{\eta(s)}^0 + X_{\eta(s)}^0) |^p]ds
\\ &\leq C_5 \frac{\epsilon^p}{n^{p/2}} + C_6 \int_0^t G_1(s)ds.
\end{align*}
This finishes the proof of Theorem \ref{firsttheorem}.
\end{proof}

\begin{proof}[Proof of Theorem \ref{secondtheorem}]
Similarly to the proof of Theorem \ref{firsttheorem}, we shall show that for 
$G_2(t) := E[\sup_{0\leq s \leq t} |X_s^{\epsilon} - \hat{Y}_t^{\epsilon, (n)}|^p]$, 
\begin{align*}
G_2(t) \leq C_1\frac{\epsilon^p}{n^p} + C_2 \int_0^t G_2(s)ds.
\end{align*}

Now we consider the decomposition 
\[
X_t^{\epsilon} - \hat{Y}_t^{\epsilon, (n)} = \tilde{e}^\epsilon(t) + \sum_{i=1}^3 \hat{e}_i^\epsilon(t), 
\]
where 
\begin{align*}
\tilde{e}^\epsilon(t) =& \int_0^t (b(X_{\eta(s)}^\epsilon, \epsilon) - b(\hat{X}_{\eta(s)}^\epsilon, \epsilon) 
+ b(\hat{X}_{\eta(s)}^0, 0) - b(X_{\eta(s)}^0, 0)) ds 
\\ & + \ \int_0^t (\sigma(X_{\eta(s)}^\epsilon, \epsilon) - \sigma(\hat{X}_{\eta(s)}^\epsilon, \epsilon) 
+ \sigma(\hat{X}_{\eta(s)}^0, 0) - \sigma(X_{\eta(s)}^0, 0)) dB_s 
\\ & + \int_0^t \int_{\eta(s)}^s ( \sigma\sigma'(X_{\eta(r)}^\epsilon, \epsilon) - \sigma\sigma'(\hat{X}_{\eta(r)}^\epsilon, \epsilon) 
\\ & \ \ \ \ + \sigma\sigma'(\hat{X}_{\eta(r)}^0, 0) - \sigma\sigma'(X_{\eta(r)}^0, 0) )dB_r dB_s, 
\end{align*}
and 
\begin{align*}
\hat{e}_1^\epsilon(t) &= 
\int_0^t \int_{\eta(s)}^s ( b\sigma'(X_{r}^\epsilon, \epsilon) - b\sigma'(X_{r}^0, 0) )dr dB_s, 
\\ \hat{e}_2^\epsilon(t) &= 
\int_0^t \int_{\eta(s)}^s ( \frac{1}{2}\sigma^2\sigma''(X_{r}^\epsilon, \epsilon) - \frac{1}{2}\sigma^2\sigma''(X_{r}^0, 0) )dr dB_s, 
\\ \hat{e}_3^\epsilon(t) &= 
\int_0^t \int_{\eta(s)}^s ( \sigma\sigma'(X_{r}^\epsilon, \epsilon) - \sigma\sigma'(X_{\eta(r)}^\epsilon, \epsilon) 
+ \sigma\sigma'(X_{\eta(r)}^0, 0) - \sigma\sigma'(X_{r}^0, 0) )dB_r dB_s.
\end{align*}
By a similar manner as in the proof of Theorem \ref{firsttheorem}, we can also obtain 
\[
E[\sup_{0\leq s \leq t } (|\tilde{e}^\epsilon(s)|^p + \sum_{i=1}^3 |\hat{e}_i^\epsilon(s)|^p) ] 
\leq C_1\frac{\epsilon^p}{n^p} + C_2 \int_0^t G_2(s)ds.
\]
Indeed, compared with Theorem \ref{firsttheorem}, the reason why we can get the rate $n^{-p}$ above 
is due to the additional integrals $\int_{\eta(s)}^s \cdot \ dr$ or 
$\int_{\eta(s)}^s \cdot \ dB_r$ inside the error terms.
\end{proof}

\section{Proof of Theorem \ref{thm:mlmc}}\label{sec:proofmlmc}
Throughout this section, we use the following notations without confusion:
\begin{itemize}
\item $f(X_T^\epsilon) \equiv f( (X_T^\epsilon)^{(1)})$, $f(\bar{X}_T^\epsilon) \equiv f((\bar{X}_T^\epsilon)^{(1)})$.
\item $\|f\|_{\mathrm{Lip}} := \inf\{ K\geq 0 : |f(x) -f(y)| \leq K |x-y|, \mbox{ for all } x, y \in \R \}$.
\item $\|f\|_{\mathrm{TV}} := \sup_{-\infty < x_0 < \cdots < x_m < \infty} \sum_{j=1}^m |f(x_j)-f(x_{j-1})|$.
\end{itemize}
We say $f$ has bounded variation in $\R$ if $\|f\|_{\mathrm{TV}}<\infty$.

The following lemma plays a crucial role in the proof of the theorem.
\begin{Alem}[Avikainen \cite{A09}, Theorem 2.4]\label{lemC1}
Let $X$ and $\hat{X}$ be real valued random variables with $X, \hat{X} \in L^p$ $(p\geq 1)$. 
In addition, suppose $X$ has a bounded density. 
Then for any function $f$ of bounded variation in $\R$ and $q \geq 1$, there exists a constant $C > 0$ 
depending on $p, q$, and the essential supremum for a density of $X$ such that
\[
\|f(X)-f(\hat{X})\|_q \leq C \|f\|_{\mathrm{TV}} \|X - \hat{X}\|_p^{\frac{p}{q(p+1)}}.
\]
\end{Alem}

By the next lemma, we can obtain an approximation sequence of the payoff $f$.
\begin{Alem}\label{lemC2}
Let $f$ be a bounded Lipschitz continuous function whose weak derivative has bounded variation in $\R$. 
Then there exists a sequence $(f_j)_{j \geq 1} \subset C^1(\R)$ such that 
\begin{align*}
& \| f - f_j \|_\infty \rightarrow 0, \ \mbox{ as } j \rightarrow \infty, 
\\ & \| f_j' \|_\infty \leq \| f \|_{\mathrm{Lip}} \ \mbox{ for all } j \geq 1, 
\\ & \| f_j' \|_{\mathrm{TV}} \leq \| f' \|_{\mathrm{TV}} \ \mbox{ for all } j \geq 1.
\end{align*}
\end{Alem}

\begin{proof}
The approximate sequence can be constructed by mollifier convolutions 
$f_h := (f * \phi_h)$, that is, $\phi_h := \frac{1}{h}\phi(\frac{x}{h})$ 
with the conditions (i) $\phi \in C^\infty$, (ii) $\mathrm{supp}(\phi) \subset \{|x| \leq 1\}$, 
(iii) $\phi \geq 0$, and (iv) $\int_\R \phi(x) dx = 1$.
\end{proof}

\begin{proof}[Proof of Theorem \ref{thm:mlmc}]
Assume that $f$ is a bounded $C^1$ function whose derivative has bounded variation. 
As seen in Proposition \ref{MLMCsmooth}, note that
\begin{align*}
& \| f(X_T^\epsilon) - f(\bar{X}_T^{\epsilon}) + f(\bar{X}_T^{0}) - f(X_T^0) \|_2
\\ &\leq 
\| \int_0^1 (f'(\rho X_T^\epsilon + (1-\rho)\bar{X}_T^\epsilon) 
- f'(\rho X_T^0 + (1-\rho)\bar{X}_T^0)) d\rho 
\cdot ( (X_T^\epsilon)^{(1)} - (\bar{X}_T^\epsilon)^{(1)} )\|_2
\\ & \ \ \ \ + \|f\|_{\mathrm{Lip}}\| X_T^\epsilon - \bar{X}_T^{\epsilon} + \bar{X}_T^{0} - X_T^0\|_2.
\end{align*}
The final line is bounded by $\|f\|_{\mathrm{Lip}} \times O(\epsilon / n^{1/2})$, 
and thus we turn to focus on the estimate for the second line. 
The second line is bounded by 
\begin{align*}
& \| \int_0^1 (f'(\rho X_T^\epsilon + (1-\rho)\bar{X}_T^\epsilon) 
- f'(\rho X_T^0 + (1-\rho)\bar{X}_T^0)) d\rho\|_{2p} 
\|X_T^\epsilon - \bar{X}_T^\epsilon\|_{2q}
\\ & \leq \frac{C_{p,T,x_0}}{n^{1/2}}  \int_0^1 \| (f'(\rho X_T^\epsilon + (1-\rho)\bar{X}_T^\epsilon) 
- f'(\rho X_T^0 + (1-\rho)\bar{X}_T^0)) \|_{2p} d\rho
\end{align*}
for any $p, q > 1$ such that $1/p + 1/q = 1$. Now using Lemma \ref{lemC1}, we have 
\begin{align*}
& \int_0^1 \| (f'(\rho X_T^\epsilon + (1-\rho)\bar{X}_T^\epsilon) 
- f'(\rho X_T^0 + (1-\rho)\bar{X}_T^0)) \|_{2p} d\rho 
\\ & \leq C_{p,r} \|f' \|_{\mathrm{TV}} \int_0^1 \| (\rho X_T^\epsilon + (1-\rho)\bar{X}_T^\epsilon) 
- (\rho X_T^0 + (1-\rho)\bar{X}_T^0) \|_{2p}^{\frac{r}{2p(r+1)}} d\rho
\\ & \leq C_{p,r,T,x_0}\|f' \|_{\mathrm{TV}} \ \epsilon^{\frac{r}{2p(r+1)}}.
\end{align*}
To obtain the result, we choose small $p > 1$ and large $r \geq 1$ such that $\frac{r}{2p(r+1)} > \frac{1}{2} - \delta$.

Finally, for general $f$, consider $f_K := (f \wedge K) \vee (-K)$ for $K > 0$ as a first approximation, 
and apply Lemma \ref{lemC2} to $f_K$. Then we obtain the desired result by taking the limit. 
\end{proof}

\section*{Acknowledgement}
The authors would like to thank Arturo Kohatsu-Higa and Akihiko Takahashi for their helpfull comments. 
The first author was supported by JSPS Research Fellowships for Young Scientists, and 
this work was supported by JSPS KAKENHI Grant Number 12J03138. 


\end{document}